\documentclass[11pt,twoside,a4paper]{article}
\usepackage{mathtools}
\usepackage{amsmath,amssymb,amsthm}
\usepackage{hyperref}
\usepackage{graphicx,psfrag}
\usepackage{nicefrac}
\usepackage{tikz}

\newtheorem{thm}{Theorem}[section]
\newtheorem{dfn}[thm]{Definition}
\newtheorem{tad}[thm]{Theorem and Definition}
\newtheorem{dat}[thm]{Definition and Theorem}

\theoremstyle{definition}

\theoremstyle{remark}
\newtheorem{rem}[thm]{Remark}

\newcommand{\R}{\mathbb{R}}
\newcommand{\Z}{\mathbb{Z}}

\newcommand{\C}{\mathbb{C}}

\renewcommand{\phi}{\varphi}

\newcommand{\charact}{\operatorname{char}}
\newcommand{\determ}{\operatorname{det}}

\begin{document}

\title{Design of Ciphers based on the Geometric Structure of the M\"obius Plane}
\author{Christoph Capellaro\footnote{EY EMEIA Financial Services - Cyber Security}}

\date{}
\maketitle

 \begin{abstract}
 	Till now geometric structures don't play a major role in cryptography. Gilbert, MacWilliams and Sloane \cite{GMS74} introduced an authentication scheme in the projective plane and showed its perfectness in the sense of Shannon \cite{Sha49}. In this paper we will show that this authentication scheme also fulfills the requirement of completeness according to Kam and Davida \cite{KD79} and we will extend the application of geometric structures in cryptography by introducing an encryption scheme in the M\"obius plane. We will further examine its properties, showing that it also fulfills the requirement of completeness and Shannon's requirement of perfectness in first approximation. The results of this paper can be used to define similar encryption schemes in the circle geometries of Laguerre and Minkowski.
 \end{abstract}


\medskip

\section{Introduction}	
 	A cryptographic transformation can be understood as an incidence relation, whereby messages ${\text m}$ and ciphertexts ${\text c}$ are represented as points. The cryptographic transformation $f$ that maps ${\text m}$ to ${\text c}$ is then described by a geometric object that incises with these two points. In this paper we will design new encryption transformations with the circle geometry of the M\"obius plane and analyze their properties.
 	
 Quite often the security of cryptographic algorithms is dependent on either the fact that no efficient algorithm is known to solve a certain mathematical problem or they use mappings that show very little structure. Examples for the former are asymmetric algorithms like RSA or the discrete logarithm. The latter includes e.g. symmetric algorithms like Feistel ciphers. Little research has been done on the applicability of geometric structures in cryptography. Gilbert, MacWilliams and Sloane showed that an authentication scheme can be defined in the projective plane (\cite{GMS74} and \cite{Beu88}). In the case that this projective plane is defined over a finite field, this authentication scheme is perfect in the sense of Shannon, i.e. the uncertainty of the plaintext after observing the ciphertext is equal to the a-priori uncertainty about the plaintext \cite{Sha49}. 
 	
 With the use of other incidence relationships, in this case based upon the circle geometry of the M\"obius plane, we will demonstrate that the definition of encryption transformations is possible. The basic property of circles, being that three points incise with a circle, allow to associate one point with a message, another point with the ciphertext and still have a degree of freedom for a secret key. We will analyze encryption methods in the M\"obius plane and will be able to show that these encryption methods fulfill the criteria of completeness in the case that the geometry is defined over a field of characteristic 2. Completeness according to Kam and Davida \cite{KD79} means that, assumed that input and output of a transformation are represented as bit vectors, there is at least one input vector for which a change in the $i$-th bit results in a change of the $j$-th bit of the output vector for arbitrary $i$ and $j$.  Encryption transformations based on M\"obius are in first approximation perfect. Strategies for key selection are introduced that provide optimal key length for a perfect encryption method, meaning that the size of the key does not exceed the size of the message, as long as padding is ignored.
 
 {\it Acknowledgment.} The author would like to express great gratitude to Helmut Karzel. Without the insights in circle geometries that he shared, this paper would not have been possible. Furthermore thank goes to Mariia Denysenko, who with her inspiring attitude and due care of technical details supported the completion of this paper.
 
\section{Cryptographic Transformations and Cryptographic Schemes} 
\begin{dfn}\label{dfn:cryptographic_transformation} Let $M$ and $C$ be sets,
	where $M$ is called a set of messages and $C$ a set of ciphertexts. If there is a nonempty set of functions $F$ of the form $F:M\to C$ with the property that every $f\in F$ is reversible, then the elements $f\in F$ are called {\em cryptographic transformations} and $(M,C,F)$ is called a {\em cryptographic scheme}.
\end{dfn}

\begin{tad}\label{tad:countable_cryptographic_scheme} A cryptographic scheme $(M,C,F)$ is called a {\em countable infinite cryptographic scheme}, if the sets $M$, $C$ and $F$ are
	countable infinite. If the set of ciphertexts $C'$ of a cryptographic scheme $(M',C',F')$ is finite then	$(M',C',F')$ is called a {\em finite cryptographic scheme}.
\end{tad}
	
\begin{proof} 
	The finiteness of the sets $M'$ and $F'$ follow immediately from the finiteness of $C'$ and from the fact that all cryptographic transformations $f'\in F'$ are injective.
\end{proof}	

\section{Properties of Cryptographic Schemes}
\begin{dfn}\label{dfn:a-priori_and_a-posteriori_probabilities} If $(M,C,F)$ is a cryptographic 	scheme, then the probability of the occurrence of a message $m\in M$ is denoted with $\mu  (m)=P(m)$ and is called {\em a-priori probability} of the message $m$. Similarly the probability of the occurrence of the message $m$ under the condition that $m$ is mapped to $c$ by any $f\in F$ is denoted with $\nu (m,c)=P_{c=f(m)}(m)$ and is
	called {\em a-posteriori probability} of the message $m$ for a given ciphertext $c$.
\end{dfn}	
	
\begin{rem} In the case that a cryptographic system $(M,C,F)$ is finite, relative frequencies can be used to calculate the a-priori and a-posteriori probabilities. Then $\mu (m)=\frac{H(m)}{|M|}$ and $\nu (m,c)=\frac{\big|\{f\in F:c=f(m)\}\big|}{\sum_{m\in M}{\big|{f\in F:c=f(m)}\big|}}$. Here $H(m)$ means the frequency of the occurrence of the message $m$.
\end{rem}

\begin{dfn}\label{dfn:perfect_cryptographic_scheme}Let $(M,C,F)$ be a cryptographic scheme for
	which every $c\in C$ is a possible ciphertext and let $\mu$ and $\nu$ be its a-priori and a-posteriori probabilities.	Then $(M,C,F,\mu ,\nu )$ is called {\em perfect according to Shannon \cite{Sha49}} as long as $\mu (m)=\nu (m,c)$ for any $m\in M$ and for any $c\in C$. 
\end{dfn}

\begin{rem} If $(M,C,F)$ is finite and if $\mu (m)=\mu (m_0)$ for any $m\in M$ then $(M,C,F,\mu ,\nu )$ is perfect as long as $\frac{1}{|M|}=\frac{\big|\{f\in F:c=f(m)\}\big|}{\sum_{m\in M}{\big|\{f\in F:c=f(m)\}\big|}}$ for any $m\in M$ and for any $c\in C$.
\end{rem}
	
Another well-known property for cryptographic schemes is the completeness according to Kam and Davida which is defined for cryptographic schemes based on binary strings.	
	
\begin{dfn}\label{dfn:complete_cryptographic_scheme} Let $(M,C,F)$ be a cryptographic
	scheme with $M=\Z_2^r$ and $C=\Z_2^s$. Then the cryptographic transformation $f\in F$ is called {\em complete according to Kam and Davida \cite{KD79}}, if there is at least one message $m_0= m_{0_1},m_{0_2},...,m_{0_r}\in M$ for every pair of indices $i\le r$ and $j\le
	s$, where a change in the $i$-th bit of $m_0\in M$ results in a change of the $j$-th bit of $c_0=c_{0_1},c_{0_2},...,c_{0_s}=f(m_0)\in C$. A cryptographic scheme $(M,C,F)$ that consists exclusively of	complete cryptographic transformations is called a {\em complete cryptographic scheme}.
\end{dfn}	
	
\section{Authentication Schemes}
\begin{dfn}\label{dfn:authentication_scheme}
	Let $(M,C,A)$ be a cryptographic scheme. Let $\mu$ further be a	measure on $C$ that fulfills the condition that $\mu (a(C))\ll \mu (C)$ for every $a\in A$. Then $(M,C,A)$ is
	called {\em authentication scheme}, the ciphertexts $c\in C$ are called {\em authenticated messages} and the	transformations $a\in A$ are called {\em authentications}.
\end{dfn}	
	
\begin{rem}
	The authentications $a\in A$ do not necessarily need to be injective. An authenticated
	message $c\in C$ is verified with a function $\nu :C\to \{true,false\}$ with $\nu (c)= true$, if $c\in a(M)$ and $\nu (c)=false$, if $c\notin a(M)$.
\end{rem}
	
So an authenticated message $c\in C$ is accepted as authentic as soon as $c\in a(M)$. This is the motivation for the condition on the authentications $a\in A$ introduced in Definition~\ref{dfn:authentication_scheme}, which ensures that only a small
subset of authenticated messages is accepted as authentic.
	
\begin{dfn}\label{dfn:cartesian_authentication_scheme}
	An authentication scheme $(M,C,A)$ is called {\em Cartesian}, if any $a\in A$ is injective and $a_1(M)\cap a_2(M)=\varnothing$ for any $a_1,a_2\in A$ and $\bigcup_{a\in A} a(M)=C$.	
\end{dfn}

For a Cartesian authentication scheme based on countable sets $C$ and $A$ counting can be defined as the measure $\mu$ on $C$ and $A$. Then the theorem of Gilbert, MacWilliams and Sloane \cite{GMS74} can be written as follows.
	
\begin{thm}\label{thm:gilbert_mcwilliams_sloane}
	Let $n_0=\nu(A)$ be the number of authentications in a	Cartesian authentication scheme $(M,C,A)$ and let $c\in C$ be a randomly chosen authenticated message, then the following holds:
	$$\nu (\{a\in A:\exists m\in M \text{\em  and } c=a(m)\})\geq \sqrt{n_0}$$
\end{thm}

A proof of theorem~\ref{thm:gilbert_mcwilliams_sloane} for an authentication scheme based on projective planes which has been introduced in \cite{GMS74} is given below, a general proof can be found e.g. in \cite{BR92}.

\begin{dfn}\label{dfn:perfect_authentication_scheme}
	A Cartesian authentication scheme $(M,C,A)$ with $n_0=\nu (A)$ for which theorem~\ref{thm:gilbert_mcwilliams_sloane} is fulfilled sharply, is called {\em perfect}.
\end{dfn}
	
\section{Authentication Schemes Based on Projective Planes}	
In this section we use some well-known facts about projective planes. The corresponding proofs can be found in the literature, e.g. in \cite{KK88} or \cite{BR92}.

\begin{dfn}\label{dfn:projective_plane}
	An incidence space $(\boldsymbol{P},\boldsymbol{L})$ of points $\boldsymbol{P}$ and lines $\boldsymbol{L}$ and at least 3
	different points in $\boldsymbol{P}$ is called a {\em projective plane}, if each line consists of at least 3 points and any two lines of $\boldsymbol{L}$ intersect in exactly one point.
\end{dfn}

\begin{dfn}\label{dfn:finite_projective_plane}
	A projective plane is {\em finite}, if its set of points is finite.
\end{dfn}
	
\begin{rem}	
	All lines in a finite projective plane have the same number of points. A projective plane
	with lines that consist of $q+1$ points is called a projective plane of order $q$.
\end{rem}
	
\begin{dat}\label{dat:perfect_authentication_scheme_in_a_finite_projective_plane}
	Let $(\boldsymbol{P},\boldsymbol{L})$ be a finite projective plane and ${\text {\em L}}_0\in \boldsymbol{L}$ be a line in $(\boldsymbol{P},\boldsymbol{L})$. An authentication scheme $(M,C,A)$ can be defined in $(\boldsymbol{P},\boldsymbol{L})$ as follows: $M$ is the set of all points on ${\text {\em L}}_0$ and $C$ the set of all lines of $\boldsymbol{L}\setminus {\text {\em L}}_0$. In order to describe the set of authentications $A$ we define $K$ as the set of points in $\boldsymbol{P}\setminus {\text {\em L}}_0$. The authentications $a\in A$ are then defined as follows: $A:M\times K\to C$  and $a({\text {\em m}},{\text {\em k}})=\overline{{\text {\em m}},{\text {\em k}}}$ for given ${\text {\em m}}\in M$ and ${\text {\em k}}\in K$, where ${\text {\em m}}$ and ${\text {\em k}}$ are considered as points of the projective plane. An authenticated message ${\text {\em C}}\in C$ which has been authenticated using a given ${\text {\em k}}_0\in K$ is verified as true, if ${\text {\em k}}_0\in {\text {\em C}}$ and false otherwise. Here again ${\text {\em k}}_0$ is considered a point of the projective plane and ${\text {\em C}}$ a line in the projective plane. This authentication scheme $(M,C,A)$ is perfect according to definition~\ref{dfn:perfect_authentication_scheme}.
\end{dat}

\begin{proof}
	Let $q$ be the order of $(\boldsymbol{P},\boldsymbol{L})$. Then $(\boldsymbol{P},\boldsymbol{L})$ consists of $q^2+q+1$ points and as many lines. There are $q+1$ points on each line and $q+1$ lines run through each point in $(P,L)$. Thus $|K|=q^2$, $|M|=q+1$ and the number of possible authentications for a given message ${\text m}\in M$ is $q$.
	
	We distinguish between the two cases that no authenticated message is known and that one such authenticated message ${\text C}\in C$ is known.
	
	To analyze the first case we select a message ${\text m}$ in the set of messages $M$ and create the authenticated message ${\text C}\in C=\boldsymbol{L}\setminus {\text L}_0$. Again we consider the representations of ${\text m}$, ${\text C}$ and ${\text L}_0$ as point and lines in the projective plane. There are $q$ possible authentications, since $q$ lines different from ${\text L}_0$ run through the point ${\text m}\in {\text L}_0$. Each of these $q$ lines has $q$ different points outside ${\text L}_0$. So there are $q^2$ possibilities for the point ${\text k}_0\in K$ which represents the key ${\text k}_0$ used for the authentication. Thus the chance that ${\text C}$ is verified as a true authenticated message equals to $\frac{q}{q^2}=\frac{1}{q}$.
	
	In the second case a message ${\text m}_0$ and its corresponding authenticated message ${\text C}_0=\overline{{\text m}_0,{\text k}_0}$ are known.	Again we are looking at the representations of ${\text m}_0$, ${\text k}_0$ and ${\text C}_0$ as points and line in the projective plane. Another ${\text m}\in {\text L}_0$ with ${\text m}\neq {\text m}_0$ and ${\text C}\in \boldsymbol{L}\setminus {\text L}_0$ with ${\text C}\neq {\text C}_0$ can be selected. The lines $\overline{{\text m},{\text k}}$ and $\overline{{\text m}_0,{\text k}_0}$ intersect in the point ${\text k}\in \overline{{\text m}_0,{\text k}_0}$. The probability that ${\text k}={\text k}_0$ is $\frac{1}{q}$.
\end{proof}

Figure~\ref{fig:authentication_projective_plane} visualizes the authentication scheme of definition~\ref{dat:perfect_authentication_scheme_in_a_finite_projective_plane}.

\begin{figure}
\begin{center}
\begin{tikzpicture}
	\draw (0,0) --(6,0) node [below] {$L_0$};
	\draw (1,-1) node [below] {${\text C}:=\overline{{\text m},{\text k}}$}--(4.5,2.5);
	\draw [fill] (4,2) circle [radius=0.1] node [right] {$\hphantom {M}{\text k}$};
	\draw [fill] (2,0) circle [radius=0.1] node [below] {$\vphantom {M}{\text m}$};
\end{tikzpicture}
\caption{Authentication in the projective plane.}\label{fig:authentication_projective_plane}
\end{center}
\end{figure}
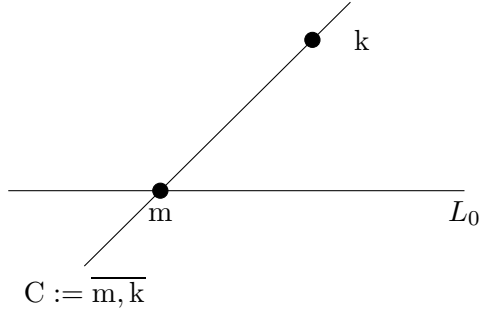
	
\begin{rem}
	The points ${\text k}\in K$ which have been used in definition~\ref{dat:perfect_authentication_scheme_in_a_finite_projective_plane} to construct the authenticated	messages are also referred to as keys.
\end{rem}

\begin{rem}
	It’s obvious that in a practical application of this perfect authentication scheme each key can be used only once.
\end{rem}
	
\begin{rem}
	Beutelspacher and Rosenbaum introduced a generalization of this perfect authentication
	scheme in the projective space \cite{BR92}.
\end{rem}
	
In this paper we continue the examination of the cryptographic properties of the authentication scheme introduced by Gilbert and MacWilliams with following theorem.

\begin{thm}
	An authentication scheme in a finite projective plane over a filed $\mathbf{F}$  with $\charact\mathbf{F}=2$ as introduced in definition~\ref{dat:perfect_authentication_scheme_in_a_finite_projective_plane} is complete in the sense of definition~\ref{dfn:complete_cryptographic_scheme}.
\end{thm}

\begin{proof}
	We consider an authentication scheme $(M,C,A)$ in the projective plane $(\boldsymbol{P},\boldsymbol{L})$ over the field $\Z_2^n$ and we introduce homogeneous coordinates ${\text x}=(\Z_2^n)^*\left (\begin{array}{c}{\text x}_1 \\ {\text x}_2 \\ {\text x}_3 \end{array} \right )$ for the points in $(\boldsymbol{P},\boldsymbol{L})$. Here and in the following the field $\mathbf{F}^*$ shall denote the field $\mathbf{F}$ without its zero element. We assume that the set of messages $M$ and the key ${\text k}$ are in the finite and that ${\text m}_1\neq {\text k}_1$. Otherwise the projective plane shall be transformed accordingly. 
	
	Since $M$ is in the finite, its points can be represented as follows:
	
	\begin{equation}\label{messages}
M=\big \{{\text x}\in \boldsymbol{P}:s{\text x}_1-{\text x}_2+t=0,s,t\in \Z_2^n,{\text x}_3=1 \big \}
	\end{equation}
	
There are $q=2^n$ points of $M$ in the finite, which are uniquely determined by the coordinate ${\text x}_1$ in~(\ref{messages}). 

Every line of $C$ passes the point ${\text k}$, the key. Since ${\text k}$ is in the finite, it can be represented with ${\text k}=({\text k}_1,{\text k}_2,1)$. Then the lines of $C$ are represented by following condition:

\begin{equation}\label{authenticated_messages}
\begin{split}
u{\text x}_1-{\text x}_2+{\text k}_2-u{\text k}_1=0 &\Leftrightarrow \\
u({\text x}_1-{\text k}_1)-{\text x}_2+{\text k}_2=0&,u\in \Z_2^n,{\text x}_3=1
\end{split}
\end{equation}

The $q=2^n$ lines of $C$ are in the finite. They can be uniquely identified with the parameter $u$ in~(\ref{authenticated_messages}). 

The authenticated message ${\text C}\in C$ is determined as that line in $C$ that passes the point ${\text m}$ which represents the message. Thus we obtain following condition for ${\text C}$:

\begin{equation}\label{authenticated_message}
u({\text x}_1-{\text k}_1)-s{\text x}_1-t+{\text k}_2=0
\end{equation}

Hence, we can refer to the authentication $a:(M,K)\to C$ as a function $a_{\text k}:\Z_2^n\to \Z_2^n$ that maps the coordinate ${\text x}_1\in \Z_2^n$ to the parameter $u\in \Z_2^n$:

\begin{equation*}
u=a_{\text k}({\text x}_1):=\frac{s{\text x}_1+t-{\text k}_2}{{\text x}_1-{\text k}_1}
\end{equation*} 

We want to look at the coordinates and parameters of $(\boldsymbol{P},\boldsymbol{L})$ as binary vectors in $\Z_2^n$. Furthermore $i,j$ shall be indices with $i,j\in \{1,2,...,n\}$ and $e_i,e_j$ the corresponding unit vectors in $\Z_2^n$. 

In order to show the completeness of the authentication $a$, we need to find a point ${\text m}'$ in a way that $a({\text m}',{\text k})={\text C}'$ when $a({\text m},{\text k})={\text C}$. Here ${\text m}'$ is obtained when ${\text x}_1$ is replaced by ${\text x}_1+e_i$ and ${\text C}'$ is obtained when $u$ is replaced by $u+e_j$ in~(\ref{authenticated_message}):

\begin{equation}\label{shifted_authenticated_message}
(u+e_j)({\text x}_1+e_i-{\text k}_1)-s({\text x}_1+e_i)-t+{\text k}_2=0
\end{equation}

Parameters $s,t\in \Z_2^n$ can be found to fulfill equations~(\ref{authenticated_message}) and~(\ref{shifted_authenticated_message}).

\end{proof}

\section{Brief Introduction of the Circle Geometry of the M\"obius Plane}
\begin{dfn}\label{dfn:moebius_plane}
	An incidence structure $(\boldsymbol{S},\boldsymbol{X})$ with a set of points $\boldsymbol{S}$ and a set of circles $\boldsymbol{X}$ is called {\em M\"obius plane}, if it satisfies following properties:
\begin{itemize}
	\item[{\em (M1)}] $\forall {\text {\em a}},{\text {\em b}},{\text {\em c}}\in \boldsymbol{S}$ with ${\text {\em a}}\ne {\text {\em b}}\ne {\text {\em c}}\ne {\text {\em a}}\  {\exists}_1 {\text {\em A}}\in \boldsymbol{X}$. Following writing convention will be used in this context: $({\text {\em a}},{\text {\em b}},{\text {\em c}})^\chi :={\text {\em A}}$.
    \item[{\em (M2)}] Touch axiom: $\forall {\text {\em A}}\in \boldsymbol{X}, \forall {\text {\em a}},{\text {\em b}}\in \boldsymbol{S}$ with ${\text {\em a}}\in {\text {\em A}}$ and ${\text {\em b}}\in \boldsymbol{S}\setminus {\text {\em A}}\  {\exists}_1{\text {\em B}}\in \boldsymbol{X}: {\text {\em a}},{\text {\em b}}\in {\text {\em B}}$ and ${\text {\em A}}\cap {\text {\em B}}=\{{\text {\em a}}\}$.
    \item[{\em (M3)}] $|\boldsymbol{S}|\geq 4$ and there are four points in $\boldsymbol{S}$ which are different from each other and which do not coincide with a common circle. 
\end{itemize}
\end{dfn}
	
	By a stereographic projection of the spherical surface on the Euclidean plane, the M\"obius geometry of the spherical surface is mapped on the geometry of the circles and lines of the Euclidean plane. If one extends the Euclidean plane with a distant point, then the straight lines can be regarded as circles through this point. Figure~\ref{fig:moebius_plane} illustrates this procedure.
	
	\begin{figure}
		\begin{center}
			\begin{tikzpicture}
			\draw (4.5,5) circle [radius=2];
			\draw (4.5,5) ellipse(2cm and 0.8cm);
	\draw (0,0) --(9,1);
	\draw (9,1) --(8,8);
	\draw (8,8) --(1,7);
	\draw (1,7) --(0,0);
	\draw (4.5,6.8) node [below right] {${\text n}$}--(3.3,1.8) node [below right] {${\text p'}$};
	\draw [fill] (4.5,6.8) circle [radius=0.1]; 
	\draw [fill] (4.5,3.2) circle [radius=0.1]; 
	\draw [fill] (3.3,1.8) circle [radius=0.1]; 
	\draw (4.5,6.8) to [out=210,in=220] (4.5,3.2);
	\draw (2.7,1.1) --(6.9,6);
	\draw [fill] (3.72,3.6) circle [radius=0.1]; 
			\end{tikzpicture}	
				\caption{M\"obius plane.}\label{fig:moebius_plane}
		\end{center}
	\end{figure}
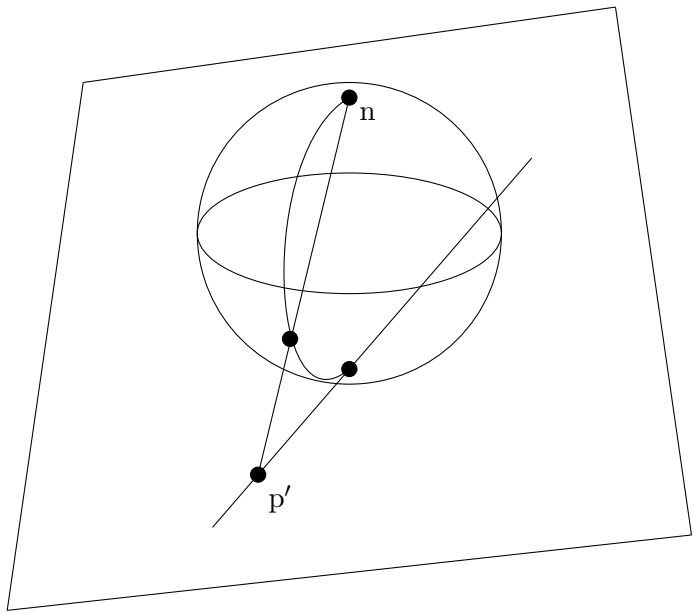
	
In the context of this paper we will only consider M\"obius planes that are defined over separable quadratic field extensions $(\mathbf{G},\mathbf{F})$ of finite fields $\mathbf{F}=GF(q)$. In a quadratic field extension of this kind there is exactly one involutorial automorphism $\overline{.\vphantom{G}}:\mathbf{G}\to \mathbf{G}$ that has the elements of $\mathbf{F}$ as fixpoints. We obtain the M\"obius plane over a separable quadratic field extension $(\mathbf{G},\mathbf{F})$, if we close it with the point $\infty$, $\boldsymbol{S}:=\overline{\mathbf{G}}:=\mathbf{G}\cup \{\infty\}$. M\"obius planes of this form do always have an analytical representation.

\begin{dfn}\label{dfn:derivation_of_a
_moebius_plane}
	The {\em derivation} of a M\"obius plane $(\boldsymbol{S},\boldsymbol{X})$ in a point ${\text {\em a}}\in \boldsymbol{S}$ is characterized with the pair $(\boldsymbol{S}^{\text {\em a}},\boldsymbol{X}^{\text {\em a}})$ with $(\boldsymbol{S}^{\text {\em a}}:=\boldsymbol{S}\setminus \{{\text {\em a}}\}$ and $\boldsymbol{X}^{\text {\em a}}:=\{{\text {\em C}}\setminus \{{\text {\em a}}\}:{\text {\em C}}\in \boldsymbol{X},{\text {\em a}}\in {\text {\em C}}\}$.
\end{dfn}

Every derivation of a M\"obius plane is an affine plane, as has been shown e.g. in \cite{KK88}. In the following we will show three different ways for an analytical representation of a M\"obius plane.

\section{Describing the M\"obius Plane Using Equations}
The Euclidean plane $\R^2$ can be represented by the plane of complex numbers $\C$, i.e. the analytical geometry of the pair $(\C,\R)$. A circle with center ${\text c}\in {\C}$ that runs through the point ${\text c}+{\text a},{\text a}\in \C^*$ is defined in $\C$ with following equation.

\begin{equation}\label{circle}
({\text z}-{\text c})\overline{({\text z}-{\text c})}={\text a}\cdot \overline{{\text a}\vphantom{()}}
\end{equation}

The circle has the radius $r=|{\text a}|$.

This model of the Euclidean geometry can be generalized by using a separable quadratic field extension $(\mathbf{L},\mathbf{K})$ instead of the pair $(\C,\R)$. Any such field extension has exactly one automorphism $\overline{.\vphantom{L}}:\mathbf{L}\to \mathbf{L}$ that keeps the elements of $\mathbf{K}$ unchanged. The automorphism $\overline{.\vphantom{L}}$ is involutory. The circles of this generalized model are presented with equation~(\ref{circle}). To obtain the M\"obius plane of the field pair $(\mathbf{L},\mathbf{K})$, $\mathbf{L}$ needs to be closed with a distant point $\infty$, $\overline{\mathbf{L}}:=\mathbf{L}\cup \{\infty\}$ and $\overline{\mathbf{K}}:=\mathbf{K}\cup \{\infty\}$.

\section{Describing the M\"obius Plane by the Use of Double Ratios}

It is known from the elementary function theory that circles in the Gaussian number plane can be described using double ratios. This can be generalized using a separable quadratic field extension $(\mathbf{L},\mathbf{K})$. The circle through the points ${\text a},{\text b},{\text c} \in \mathbf{L}$ consists of the points ${\text z}$ for which the double ratio

\begin{equation}
\text{ Dr}({\text a},{\text b},{\text c},{\text z}):=\nicefrac{\frac{{\text a}-{\text c}}{{\text a}-{\text z}}}{\frac{{\text b}-{\text c}}{{\text b}-{\text z}}}
\end{equation}

is a value in $\mathbf{K}$. The coresponding M\"obius plane is defined by adding the point $\infty$, i.e. $\boldsymbol{S}:=\overline{\mathbf{L}}:=\mathbf{L}\cup \{\infty\}$.

\section{Describing the M\"obius Plane using fractional linear functions}
In the M\"obius plane $(\boldsymbol{S},\boldsymbol{X})$ over a separable quadratic field extension $(\mathbf{G},\mathbf{F})$  fractional linear functions of the following form can be chosen to describe circles:
\begin{equation}
\gamma:
	\begin{cases}
		\overline{\mathbf{F}}\to \boldsymbol{X}\\
		z\mapsto \frac{az+b}{cz+d},&ad-bc\in \mathbf{G}^*, z\neq \infty, z\neq -\frac{d}{c} \mbox{, if } \frac{d}{c}\in \mathbf{F}\\
		\infty \mapsto \frac{a}{c}\\
		-\frac{d}{c}\mapsto \infty, &\mbox{if }\frac{d}{c}\in \mathbf{F}\\
	\end{cases}
\end{equation}

\section{Combinatorial Aspects of the M\"obius Plane}
In the following we will focus on M\"obius planes based on finite fields. Let $(\boldsymbol{S},\boldsymbol{X})$ be a M\"obius plane over the field $\mathbf{F}$ with $|\mathbf{F}|=:q$. Since any derivation of the M\"obius plane in a point results in an affine plane, where every line has $q$ points, all circles in the M\"obius plane have the same number of $k:=q+1$ points. By calling $\nu$ the number of points in $\boldsymbol{S}$ and $b$ the number of circles in $\boldsymbol{X}$, we find the following relationships by counting coincidences. From~(M1) we know:

\begin{equation}\label{equ:number_of_circles}
	\left(\begin{array}{r}
	\nu\\
	3\\
	\end{array}\right)=b\cdot
	\left(\begin{array}{r}
	k\\
	3\\
	\end{array}\right) 
\end{equation}

Let $x$ be the number of circles that run through two different points of $(\boldsymbol{S},\boldsymbol{X})$, because of~(M1)  there is:

\begin{equation}\label{equ:circles_through_two_points}
	x(q-1)=\nu -2
\end{equation}

Similarly it can be shown that as many circles coincide with two different points ${\text a}$ and ${\text b}$ in $(\boldsymbol{S},\boldsymbol{X})$, as there are points on one circle. Therefore let ${\text C},{\text D}\in \boldsymbol{X}$ and ${\text a},{\text b}\in {\text C}, {\text C}\cap {\text D}=\{{\text a}\}$.

	\begin{figure}
	\begin{center}
		\begin{tikzpicture}
		\draw [gray] (5,5) circle [radius=2];
		\draw (3,4.5) circle [radius=1];
		\draw (4.45,2.5) circle [radius=1.49];
		\draw [fill] (3.52,3.63) circle [radius=0.1] node [below] {\phantom{I}${\text a}$}; 
		\draw [fill] (5.76,3.15) circle [radius=0.1] node [below right] {\phantom{I}${\text b}$}; 
		\draw (5.9,2) node [right] {${\text C}$}; 
		\draw (1.5,4) node [right] {${\text D}$}; 
		
		\end{tikzpicture}	
		\caption{Circles through two points of the M\"obius plane.}\label{fig:circles_through_two_points}
	\end{center}
\end{figure}
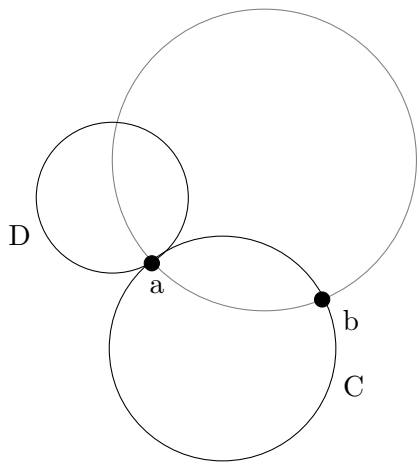

Figure~\ref{fig:circles_through_two_points} shows that any circle through ${\text a}$ and ${\text b}$, except the circle ${\text C}$ itself, intersects the circle ${\text D}$ in a point different from ${\text a}$. In connection with~(M1) and~(M2) this proves our assertion. Furthermore we can conclude from equation~(\ref{equ:circles_through_two_points}) that the number of points in the M\"obius plane is $\nu = q^2+1$. By substituting the relations for $\nu$ and $k$ we get from equation~(\ref{equ:number_of_circles}) that the number of circles in the M\"obius plane $b=q(q^2+1)$. Since the derivation of a M\"obius plane in any point ${\text a}\in \boldsymbol{S}$ results in an affine plane over the field $\mathbf{F}$ and since the affine plane has $q^2+q$ lines, as much circles of the M\"obius plane coincide with the point ${\text a}$. Furthermore there are $q$ circles different from the circle ${\text C}$ that touch ${\text C}$ in the point ${\text b}\in {\text C}$. Figure~\ref{fig:touching_circles} shows that any circle touching ${\text C}$ in the point ${\text b}$ intersects with the circle ${\text D}$ in exactly one point. Since circle ${\text D}$ has $q+1$ points, the assertion is correct.

	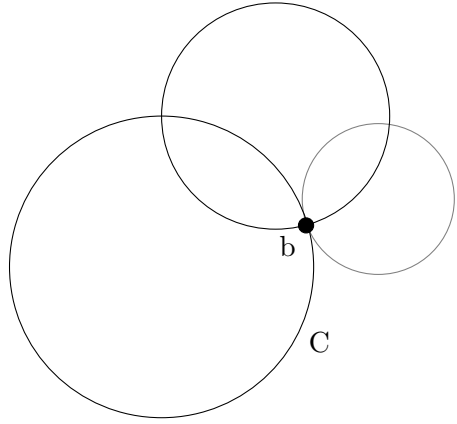
\begin{figure}
	\begin{center}
		\begin{tikzpicture}
		\draw [gray] (4.85,2.9) circle [radius=1];
		\draw (2,2) circle [radius=2];
		\draw (3.5,4) circle [radius=1.5];
		\draw [fill] (3.9,2.55) circle [radius=0.1] node [below left] {\phantom{I}${\text b}$}; 
		\draw (3.8,1) node [right] {${\text C}$}; 
		
		\end{tikzpicture}	
		\caption{Touching circles through a point of the M\"obius plane.}\label{fig:touching_circles}
	\end{center}
\end{figure}

\section{Encryption in the M\"obius Plane}
In the following we define a cipher system in a M\"obius plane $(\boldsymbol{S},\boldsymbol{X})$ over a finite field $\mathbf{F}$ with $|\mathbf{F}|=:q$. For this purpose we define

\begin{equation}
	\left( \begin{array}{r}
		\boldsymbol{S}\\
		3\\
	\end{array}\right):=\big \{({\text s}_1,{\text s}_2,{\text s}_3)\in \boldsymbol{S}:\big|\{{\text s}_1,{\text s}_2,{\text s}_3\}\big|=3\big \}
\end{equation}

as the set of triples in $\boldsymbol{S}$ that consist of three different elements.

\begin{dfn}\label{dfn:moebius_cipher}
	A cipher system $(M,C,F)$ is defined in the M\"obius plane $(\boldsymbol{S},\boldsymbol{X})$ as follows:

\begin{itemize}
	\item [{\em Messages}] $M:=\left(\begin{array}{r}
	\boldsymbol{S}\\
	3\\
	\end{array}	\right)$

\item[{\em Cipher texts}] $C:=\left(\begin{array}{r}
\boldsymbol{S}\\
3\\
\end{array}	\right)$

\item[{\em Encryption functions}]
$F:M\to C$ such that, if ${\text {\em M}}:=({\text {\em m}}_1,{\text {\em m}}_2,{\text {\em m}}_3)^\chi$ is the circle induced by the message $({\text {\em m}}_1,{\text {\em m}}_2,{\text {\em m}}_3)$, then encryption functions $f\in F$ are defined as $f:{\text {\em M}}\to {\text {\em M}}$ in the following way: 

\begin{itemize}
	\item [{\em (MC1)}]$f$ is associated with the triple $({\text {\em K}}_1,{\text {\em K}}_2,{\text {\em K}}_3)\in 
	\left(\begin{array}{r}
	\boldsymbol{X}\\
	3\\
	\end{array}	\right)$, which is called the {\em key}.
	\item[{\em (MC2)}]${\text {\em m}}_i\in {\text {\em K}}_i$ and $({\text {\em m}}_1,{\text {\em m}}_2,{\text {\em m}}_3)^\chi \cap {\text {\em K}}_i  \cap {\text {\em K}}_j=\varnothing$ for $i,j\in \{1,2,3\}$ and $i\neq j$.
\end{itemize}

\item[{\em Encryption}]
For a given message $({\text {\em m}}_1,{\text {\em m}}_2,{\text {\em m}}_3)$, an encryption function $f\in F$ and an associated key $({\text {\em K}}_1,{\text {\em K}}_2,{\text {\em K}}_3)\in 
\left(\begin{array}{r}
\boldsymbol{X}\\
3\\
\end{array}	\right)$ the ciphertext $({\text {\em c}}_1,{\text {\em c}}_2,{\text {\em }}_3):=f({\text {\em m}}_1,{\text {\em m}}_2,{\text {\em m}}_3)$ is determined by

\begin{itemize}
\item [{\em (MC3)}]
	${\text {\em c}}_i:=
	\begin{cases}
		{\text {\em K}}_i\cap ({\text {\em m}}_1,{\text {\em m}}_2,{\text {\em m}}_3)^\chi - \{{\text {\em m}}_i\} & \mbox {if } \big|{\text {\em K}}_i\cap ({\text {\em m}}_1,{\text  {\em m}}_2,{\text {\em m}}_3)^\chi\big|=2\\
		{\text {\em m}}_i & \mbox {if } \big|{\text {\em K}}_i\cap ({\text {\em m}}_1,{\text {\em m}}_2,{\text {\em m}}_3)^\chi\big|=1\\
	\end{cases}$
\end{itemize}

\item[{\em Decryption}]
Applying the encryption function $f$ to the cipher text $({\text {\em c}}_1,{\text {\em c}}_2,{\text {\em c}}_3)\in C$ yields the message.
\end{itemize}

The cipher system $(M,C,F)$ is called {\em M\"obius Cipher}.
 
\end{dfn}

Figure~\ref{fig:moebius_cipher} illustrates the encryption process of the M\"obius Cipher.

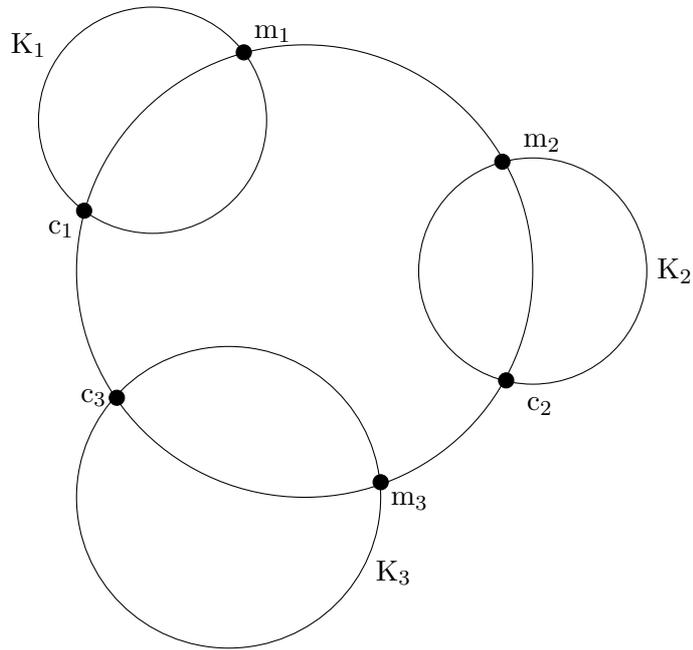
\begin{figure}
	\begin{center}
		\begin{tikzpicture}
		\draw (4,4) circle [radius=3];
		\draw (2,6) circle [radius=1.5];
		\draw (7,4) circle [radius=1.5];
		\draw (3,1) circle [radius=2];

		\draw [fill] (3.2,6.9) circle [radius=0.1] node [above right] {${\text m}_1$}; 
		\draw [fill] (1.1,4.8) circle [radius=0.1] node [below left] {${\text c}_1$}; 
		\draw [fill] (1.53,2.32) circle [radius=0.1] node [ left] {${\text c}_3$}; 
		\draw [fill] (5,1.2) circle [radius=0.1] node [below right] {${\text m}_3$}; 
		\draw [fill] (6.65,2.55) circle [radius=0.1] node [below right] {\phantom{I}${\text c}_2$}; 
		\draw [fill] (6.6,5.45) circle [radius=0.1] node [above right] {\phantom{I}${\text m}_2$}; 
		
		\draw (4.8,0) node [right] {${\text K}_3$}; 
		\draw (8.5,4) node [right] {${\text K}_2$}; 
		\draw (0,7) node [right] {${\text K}_1$}; 
		
		\end{tikzpicture}	
		\caption{Visualization of the M\"obius cipher.}\label{fig:moebius_cipher}
	\end{center}
\end{figure}

\begin{rem}
	Due to the prerequisite that the points ${\text m}_1,{\text m}_2,{\text m}_3$ of the message are different from each other in pairs, it is ensured that the circle through these three points is well defined. This allows the chosen construction of the encryption function.
\end{rem}

\begin{rem}
	The way, how the circles of the key have been chosen, ensures that the three points of the cipher text ${\text c}_1,{\text c}_2,{\text c}_3$ are also different from each other in pairs and that $({\text c}_1,{\text c}_2,{\text c}_3)^\chi = ({\text m}_1,{\text m}_2,{\text m}_3)^\chi$.
\end{rem}

\begin{rem}
	Definition~\ref{dfn:moebius_cipher} doesn't describe a constructive way to find suitable circles to determine the key. Some exemplary strategies to get such key circles are explained in \cite{Cap}.
\end{rem}
	
\begin{rem}
	In case the message is shorter than three points in the M\"obius plane, one or two points have to be added to make the M\"obius cipher applicable. This mechanism is called "padding". Both sender and receiver of encrypted messages need to be aware about the chosen padding mechanism.
\end{rem}

\begin{rem}
	If the message consists of more than three points in the M\"obius plane, the encryption with the M\"obius cipher is repeated with the next three consecutive points until the end of the message is reached. Again some padding might be required, if there are less than three points remaining at the end of the message. In order to achieve cryptographic security, for every new triple of points from the message three new circles have to be chosen as key. 
\end{rem}

\section{Cryptoanalysis of the M\"obius Cipher}
Quite often the security of cryptographic algorithms is dependent on either the fact that no efficient algorithm is known to solve a certain mathematical problem or they use mappings that show very little structure. Examples for the former are asymmetric algorithms like RSA or the discrete logarithm. The latter includes e.g. symmetric algorithms like Feistel ciphers. In both cases the security of these ciphers is based on assumptions and experience, but cannot be concluded strictly. The advantage of the M\"obius cipher that has been introduced in this paper is that it is based on well defined geometric structures and hence its cryptographic strength can be derived from known properties.

Before coming to the perfectness and completeness as introduced in definitions~\ref{dfn:perfect_cryptographic_scheme} and~\ref{dfn:complete_cryptographic_scheme}, we would like to discuss some other requirements for a cipher system. An important one is that the number of potential keys needs to be large enough to avoid successful brute force attacks. Typically $2^{128}$ is nowadays considered to be a sufficient size for the number of possible keys. The number field $\mathbf{F}$ over which the M\"obius plane $(\boldsymbol{S},\boldsymbol{X})$ has been defined needs to be large enough, to fulfill this requirement. Another useful property of a cipher system is the ratio between input and output length. If ciphertexts are longer than the original messages, the cipher system is called {\em expanding}. Again a suitable choice of the underlying number field $\mathbf{F}$ can reduce expansion. In the case that ASCII characters, which have a length of 8 bit, shall be encrypted, expansion can be kept low, if the number field $q:=|\mathbf{F}|$ is chosen in a way that $q$ is only slightly bigger than $2^{k\cdot 8}$ and $k \geq 16$. This choice would also fulfill the requirement mentioned above of providing a sufficient large key space.

To examine the properties of perfectness and completeness we use a simple strategy to select the necessary keys for the M\"obius cipher $(M,C,F)$ by choosing lines. In the M\"obius plane $(\boldsymbol{S},\boldsymbol{X})$ lines are represented by circles ${\text C}\in C$ that run through the point $\infty$. This gives us following quantities for our M\"obius cipher $(M,C,F)$ in the M\"obius plane $(\boldsymbol{S},\boldsymbol{X})$ over the field $\mathbf{F}$ with $|\mathbf{F}|=:q$. The messages ${\text m}_1,{\text m}_2,{\text m}_3$ are three points in $(\boldsymbol{S}\setminus \{\infty\})=:\overset{\circ}{\boldsymbol{S}}$, different in pairs. The key circles ${\text K}_1,{\text K}_2,{\text K}_3$ are determined by the three points ${\text k}_i\in \boldsymbol{S}\setminus \{({\text m}_1,{\text m}_2,{\text m}_3)^\chi,\infty\}$, $i\in \{1,2,3\}$. The following figure~\ref{fig:encryption_with_lines} illustrates the encryption function.

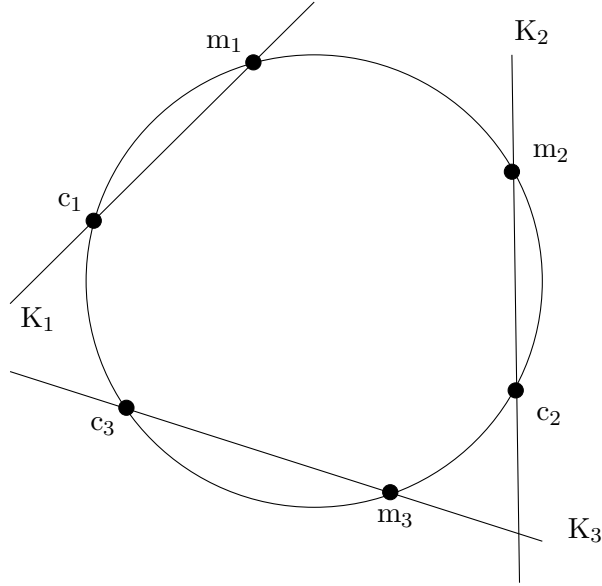
\begin{figure}
	\begin{center}
		\begin{tikzpicture}
		\draw (4,4) circle [radius=3];
		\draw (0,3.7) --(4,7.7);
		\draw (0,2.8) --(7,0.55);
		\draw (6.6,7) --(6.7,0);
		
		\draw [fill] (3.2,6.9) circle [radius=0.1] node [above left] {${\text m}_1$}; 
		\draw [fill] (1.1,4.8) circle [radius=0.1] node [above left] {${\text c}_1$}; 
		\draw [fill] (1.53,2.32) circle [radius=0.1] node [below left] {${\text c}_3$}; 
		\draw [fill] (5,1.2) circle [radius=0.1] node [below] {\phantom{I}${\text m}_3$}; 
		\draw [fill] (6.65,2.55) circle [radius=0.1] node [below right] {\phantom{I}${\text c}_2$}; 
		\draw [fill] (6.6,5.45) circle [radius=0.1] node [above right] {\phantom{I}${\text m}_2$}; 
		
		\draw (7.2,0.7) node [right] {${\text K}_3$}; 
		\draw (6.5,7.3) node [right] {${\text K}_2$}; 
		\draw (0,3.5) node [right] {${\text K}_1$}; 
		
		\end{tikzpicture}	
		\caption{M\"obius cipher using lines as keys.}\label{fig:encryption_with_lines}
	\end{center}
\end{figure}

The encryption of three points of the message requires three points as keys. Hence the length of the key equals the length of the message. Since according to Shannon the key length in a perfect encryption scheme is at least as big as the length of the ciphertext, the key length of our M\"obius cipher is minimal.

To examine the property of perfectness according to definition~\ref{dfn:perfect_cryptographic_scheme} we want to assume that all messages in $(M,C,F)$ are equally distributed. Let ${\text m}_1\neq{\text m}_2\neq{\text m}_3\neq{\text m}_1$ be three points of the message lying on the circle ${\text M}$ of the M\"obius plane and let ${\text c}_1,{\text c}_2,{\text c}_3$ be the corresponding points of the cipher text, which also incide with the circle ${\text M}$. Furthermore the probability measures $\mu$ and $\nu$ shall be defined on ${\text M}$ as explained in definition~\ref{dfn:perfect_cryptographic_scheme}. In the following we want to determine the a-priori and a-posteriori probabilities of the three points ${\text m}_1,{\text m}_2,{\text m}_3$ of the message independently. 

Lets have a look at the encryption of ${\text m}_1$. ${\text M}:=({\text m}_1,{\text m}_2,{\text m_3})^\chi$ is the circle determined by the message points. Since there are $q+1$ points on ${\text M}$, the probability of the occurrence of the message point ${\text m}_1$ is $\mu ({\text m}_1)=\frac{1}{q+1}$.

The number of possible keys that can be used to encrypt ${\text m}_1$ is given by $K_1:=\overset{\circ}{\boldsymbol{S}}\setminus{{\text M}}$, hence $|K_1|=q^2-q-1$.

The pair of points ${\text m}_1,{\text c}_1$ defines a line in $(\overset{\circ}{\boldsymbol{S}},{\boldsymbol{X}})$. If ${\text m}_1={\text c}_1$, we choose the tangent to the circle ${\text M}$ through the point ${\text m}_1$. Two different cases are examined to determine the a-posteriori probabilities: 

\begin{itemize}
	\item [${\text m}_1={\text c}_1$:] Every point of the tangent to the circle ${\text M}$ through the point ${\text c}_1$ except the points ${\text c}_1$ and $\infty$ are possible key points. Hence in this case $\big| \{f\in F: ({\text m}_1,{\text c}_1)\subset f\}\big|=q-1$.
	\item [${\text m}_1\neq {\text c}_1$:]For a given pair of message and ciphertext points ${\text m}_1,{\text c}_1$ the key can be any point on the line defined by ${\text m}_1,{\text c}_1$, with the exception of the points ${\text m}_1,{\text c}_1,\infty$. Hence $\big| \{f\in F: ({\text m}_1,{\text c}_1)\subset f\}\big|=q-2$.
\end{itemize}

We can conclude that in case of ${\text m}_1={\text c}_1$ the a-posteriori probability is  
$\nu ({\text m}_1,{\text c}_1)=\frac{q-1}{q^2-q-1}$ and in the case of ${\text m}_1\neq {\text c}_1$ it is $\nu ({\text m}_1,{\text c}_1)=\frac{q-2}{q^2-q-1}$. Our a-priori probability was $\mu ({\text m}_1)=\frac{1}{q+1}$. So with large $q$ we can state that in first approximation $\mu ({\text m}_1)\approx \nu ({\text m}_1,{\text c}_1)$.

Lets proceed with examining the encryption of the message point ${\text m}_2$. Due to the prerequisite that the tree message points ${\text m}_1,{\text m}_2,{\text m}_3$ determine a unique circle ${\text M}$ in the M\"obius plane, we know that ${\text m}_2\neq {\text m}_1$. Hence the a-priori probability equals $\mu({\text m}_2)=\frac{1}{q}$. According to condition (MC2) for the encryption functions of the M\"obius  cipher in definition~\ref{dfn:moebius_cipher} the set $K_2$ of possible keys for the encryption of ${\text m}_2$ is reduced by those key points that would map ${\text m}_2$ to ${\text m}_1$ or ${\text c}_1$. Thus in case of ${\text m}_1={\text c}_1$ we get $K_2:=K_1\setminus \{({\text m}_1,{\text m}_2,\infty)^\chi \}$ and $|K_2|=q^2-2q+1$. In  case of ${\text m}_1\neq {\text c}_1$ we get $K_2:=K_1\setminus \{({\text m}_2,{\text m}_1,\infty)^\chi ,({\text m}_2,{\text c}_1,\infty)^\chi \}$ and $|K_2|=q^2-3q+3$.

To calculate the a-posteriori probabilities of the encryption of ${\text m}_2$, we have to distinguish the cases ${\text m}_2= {\text c}_2$ and ${\text m}_2\neq {\text c}_2$ again:

\begin{itemize}
	\item [${\text m}_1= {\text c}_1$,${\text m}_2= {\text c}_2$]Every point of the tangent through the point ${\text c}_2$ at the circle $M$, except the points ${\text c}_2$ and $\infty$ are possible key points. Hence $\big| \{f\in F: ({\text m}_2,{\text c}_2)\subset f\}\big|=q-1$.
	\item [${\text m}_1= {\text c}_1$,${\text m}_2\neq {\text c}_2$] For a given pair of message and ciphertext points ${\text m}_2,{\text c}_2$ all points of the line through the points ${\text m}_2,{\text c}_2$ with the exception of ${\text m}_2,{\text c}_2,\infty$ can be keys. Thus $\big| \{f\in F: ({\text m}_2,{\text c}_2)\subset f\}\big|=q-2$.
	\item[${\text m}_1\neq {\text c}_1$,${\text m}_2= {\text c}_2$]Every point of the tangent through the point ${\text c}_2$ at the circle $M$, except the points ${\text c}_2$ and $\infty$ are possible key points. Hence $\big| \{f\in F: ({\text m}_2,{\text c}_2)\subset f\}\big|=q-1$.
	\item [${\text m}_1\neq {\text c}_1$,${\text m}_2\neq {\text c}_2$] For a given pair of message and ciphertext points ${\text m}_2,{\text c}_2$ all points of the line through the points ${\text m}_2,{\text c}_2$ with the exception of ${\text m}_2,{\text c}_2,\infty$ can be keys. Thus $\big| \{f\in F: ({\text m}_2,{\text c}_2)\subset f\}\big|=q-2$.
\end{itemize}

The following table summarizes the a-priori and a-posteriori probabilities for the encryption of the second message point:

\begin{center}
\begin{tabular}{|c|c|c|}
	\hline 
	Case & A-priori probability & A-posteriori probability \\
	\hline
	${\text m}_1= {\text c}_1$,${\text m}_2= {\text c}_2$ & $\mu({\text m}_2)=\frac{1}{q}$ & $\nu ({\text m}_2,{\text c}_2)=\frac{1}{q-1}$ \\
	\hline
	${\text m}_1= {\text c}_1$,${\text m}_2\neq {\text c}_2$ & $\mu({\text m}_2)=\frac{1}{q}$ & $\nu ({\text m}_2,{\text c}_2)=\frac{q-2}{q^2-2q+1}$ \\
	\hline
	${\text m}_1\neq {\text c}_1$,${\text m}_2= {\text c}_2$ & $\mu({\text m}_2)=\frac{1}{q}$ & $\nu ({\text m}_2,{\text c}_2)=\frac{q-1}{q^2-3q+3}$ \\
\hline
	${\text m}_1\neq {\text c}_1$,${\text m}_2\neq {\text c}_2$ & $\mu({\text m}_2)=\frac{1}{q}$ & $\nu ({\text m}_2,{\text c}_2)=\frac{q-2}{q^2-3q+3}$ \\
\hline
\end{tabular}
\end{center}

Also the second encryption step of our M\"obius cipher is in first approximation perfect, when considering large values for $q$.

Similar considerations lead to the following results for the encryption of the third message point ${\text m}_3$:

\begin{center}
	\begin{tabular}{|c|c|c|}
		\hline 
		Case & A-priori probability & A-posteriori probability \\
		\hline
		${\text m}_1= {\text c}_1$,${\text m}_2= {\text c}_2$,${\text m}_3= {\text c}_3$ & $\mu({\text m}_3)=\frac{1}{q-1}$ & $\nu ({\text m}_3,{\text c}_3)=\frac{q-1}{q^2-3q+3}$ \\
		\hline
		${\text m}_1= {\text c}_1$,${\text m}_2= {\text c}_2$,${\text m}_3\neq {\text c}_3$ & $\mu({\text m}_3)=\frac{1}{q-1}$ & $\nu ({\text m}_3,{\text c}_3)=\frac{q-2}{q^2-3q+3}$ \\
		\hline
		${\text m}_1= {\text c}_1$,${\text m}_2\neq {\text c}_2$,${\text m}_3= {\text c}_3$ & $\mu({\text m}_3)=\frac{1}{q-1}$ & $\nu ({\text m}_3,{\text c}_3)=\frac{q-1}{q^2-4q+5}$ \\
\hline
${\text m}_1= {\text c}_1$,${\text m}_2\neq {\text c}_2$,${\text m}_3\neq {\text c}_3$ & $\mu({\text m}_3)=\frac{1}{q-1}$ & $\nu ({\text m}_3,{\text c}_3)=\frac{q-2}{q^2-4q+5}$ \\
\hline
		${\text m}_1\neq {\text c}_1$,${\text m}_2= {\text c}_2$,${\text m}_3= {\text c}_3$ & $\mu({\text m}_3)=\frac{1}{q-1}$ & $\nu ({\text m}_3,{\text c}_3)=\frac{q-1}{q^2-4q+5}$ \\
\hline
${\text m}_1\neq {\text c}_1$,${\text m}_2= {\text c}_2$,${\text m}_3\neq {\text c}_3$ & $\mu({\text m}_3)=\frac{1}{q-1}$ & $\nu ({\text m}_3,{\text c}_3)=\frac{q-2}{q^2-4q+5}$ \\
\hline
		${\text m}_1\neq {\text c}_1$,${\text m}_2\neq {\text c}_2$,${\text m}_3= {\text c}_3$ & $\mu({\text m}_3)=\frac{1}{q-1}$ & $\nu ({\text m}_3,{\text c}_3)=\frac{q-1}{q^2-5q+7}$ \\
\hline
${\text m}_1\neq {\text c}_1$,${\text m}_2\neq {\text c}_2$,${\text m}_3\neq {\text c}_3$ & $\mu({\text m}_3)=\frac{1}{q-1}$ & $\nu ({\text m}_3,{\text c}_3)=\frac{q-2}{q^2-5q+7}$ \\
\hline
	\end{tabular}
\end{center}

So, also the third encryption step of the M\"obius cipher is in first approximation perfect for large values for $q$.

To show the completeness of the M\"obius cipher in the sense of definition~\ref{dfn:complete_cryptographic_scheme}, we will look at the special case of a number field $\mathbf{F}$ with $\charact\mathbf{F}=2$. The points ${\text p}\in \overset{\circ}{\boldsymbol{S}}$ can then be described in the form ${\text p}(x,y)$ and $x\in \Z_2^n$, $y\in \Z_2^n$. A message of $2n$ bits length can be understood as the point ${\text p}(x,y)$ with coordinates $x$ and $y$. The $i$-th bit of the message shall be the $i$-th position in the representation of the point ${\text p}(x,y)$, which would be the $i$-th component of the vector $x$ for $i=1,...,n$ and the $(i-n)$-th component of the vector $y$ for $i=n+1,...,2n$. 

Let ${\text m}$ be a message point and ${\text c}$ be the corresponding ciphertext point when using the key point ${\text k}$. The encryption is complete in the sense of definition~\ref{dfn:complete_cryptographic_scheme}, if there is a representation of ${\text m}$ where a change in position $j$ of ${\text c}$ is caused by a change in position $i$ of ${\text m}$, $i,j\in \{1,2,...,2n\}$. 

Two cases have to be considered. The indices $i$ and $j$ may affect the same coordinate of both ${\text m}$ and ${\text c}$. Without restricting generality we assume the $x$-coordinate. Alternatively, the first index $i$ may affect one coordinate, say the $x$-coordinate of ${\text m}$, the second index $j$ may then affect the $y$-coordinate of ${\text c}$. 

We have a look at the second case first. Let ${\text m}(x,y)$ and ${\text c}(u,v)$ be the message and ciphertext points and their coordinates. Furthermore, let $e_i$, $i\in \{1,...,n\}$ be the unit vectors in $\Z_2^n$. We transform the index $j$ to $j\rightarrow j-n$. Then ${\text m}'(x+e_i,y)$ and ${\text c}'(u,v+e_j)$ are the message and ciphertext points with changes in positions $i$ and $j$. It has to be shown that there are two lines ${\text G}$ and ${\text H}$ in the M\"obius plane $(\boldsymbol{S},\boldsymbol{X})$ over the field $\Z_2^n$ with ${\text m},{\text c},{\text k}\in {\text G}$ and ${\text m}',{\text c}',{\text k}\in {\text H}$. Without restricting generality we assume that ${\text k}=(0,0)$, otherwise the M\"obius plane can be transformed accordingly. Then our assertion is equal to 

\begin{equation}\label{assertion_completeness}
\begin{split}
	\determ({\text m},{\text c})&=0 \\
	\determ({\text m}',{\text c}')&=0
\end{split}
\end{equation}

Substituting the coordinates for ${\text m},{\text m}',{\text c},{\text c}'$ results in:

\begin{equation}\label{assertion_completeness_with_coordinates}
\begin{split}
xv-uy&=0 \\
xv+xe_j+ve_i+e_ie_j-uv&=0
\end{split}
\end{equation}

Or:

\begin{equation}
\begin{split}
xv-uy&=0 \\
xe_j+ve_i+e_ie_j&=0
\end{split}
\end{equation}

It is easy to provide coordinates for ${\text m}(x,y)$ and ${\text c}(u,v)$ that satisfy these equations.

To show the first case we again assume that the two points ${\text m}(x,y)$ and ${\text c}(u,v)$ and their coordinates as given. Let now be ${\text m}'(x+e_i,y)$ and ${\text c}'(u+e_j,v)$ be the altered points. Once again it has to be shown that condition (\ref{assertion_completeness}) is fulfilled. In analogy to the procedure shown above we reach following conditions:

\begin{equation}
\begin{split}
xv-uy&=0 \\
xv+e_iv-yu-ye_i&=0
\end{split}
\end{equation}

Which can be simplified to:

\begin{equation}
\begin{split}
xv-uy&=0 \\
e_jy+ve_i&=0
\end{split}
\end{equation}

Again it is easy to find suitable points ${\text m}$ and ${\text c}$ for all possible $i,j$.

\section{R\'{e}sum\'{e} and Outlook}
We showed that geometric structures cannot only be used to define authentication schemes but that they can also be used to design encryption functions. A cryptographic transformation has been introduced in a M\"obius plane over a finite field that is approximately perfect in the sense of Shannon \cite{Sha49} and complete according to \cite{KD79}. We think that the M\"obius cipher can be of good practical use in the area of quantum cryptography, where a constant stream of qbits is shared as key between two parties to support encryption. In comparison to the well known Vernam cipher~\cite{Ve26}, the M\"obius cipher has the advantage of not only being approximately perfect according to Shannon, but also complete according to Kam and Davida. The results of this paper can easily generalized for geometric structures of the Laguerre and Minkowski plane.

\end{document}